\documentclass[copyright]{eptcs}

\providecommand{\event}{AFL 2014} % Name of the event you are submitting to
\usepackage{mathrsfs}
\usepackage{amssymb}
\usepackage{amsmath}
\usepackage{theorem}
\usepackage{multicol}

\newcommand{\lfam}{\mathscr{L}}

\newcommand{\pcfa}{\mbox{\textrm{PCFA}}}
\newcommand{\dpcfa}{\mbox{\textrm{DPCFA}}}

\newcommand{\cpcfa}{\mbox{\textrm{CPCFA}}}
\newcommand{\rcpcfa}{\mbox{\textrm{RCPCFA}}}

\newcommand{\drcpcfa}{\mbox{\textrm{DRCPCFA}}}
\newcommand{\fdrcpcfa}[1]{\mbox{$#1$\textrm{-DRCPCFA}}}

\newcommand{\rightend}{\mathord{\vartriangleleft}}  % amssymb 

\newcommand{\oca}{\textrm{OCA}}
\newcommand{\valc}{\textrm{VALC}}
\newcommand{\invalc}{\textrm{INVALC}}

\newcommand{\dollar}{\texttt{\$}}
\newcommand{\kand}{\texttt{\&}}

\theorembodyfont{\slshape} 
\theoremstyle{plain} 
 \newtheorem{definition}{Definition}
 \newtheorem{lemma}[definition]{Lemma}
 \newtheorem{theorem}[definition]{Theorem}
 \newtheorem{corollary}[definition]{Corollary}

\theorembodyfont{\normalfont}
 \newtheorem{example}[definition]{Example}

\def\squareforqed{$\Box$}
\def\qed{\ifmmode\squareforqed\else{\unskip\nobreak\hfil%
  \penalty50\hskip1em\null\nobreak\hfil\squareforqed%
  \parfillskip=0pt\finalhyphendemerits=0\endgraf}\fi}

\newenvironment{proof}{\noindent{\textbf{Proof}\ }}{\qed\medskip}

\title{Measuring Communication in\\ Parallel Communicating Finite Automata}

\def\titlerunning{Measuring Communication in Parallel Communicating Finite Automata}
\def\authorrunning{H. Bordihn, M.~Kutrib, A.~Malcher}

\author{%
Henning Bordihn
\institute{Institut f\"ur Informatik, Universit\"at Potsdam,\\
  August-Bebel-Str.~89, 14482 Potsdam, Germany}
  \email{henning@cs.uni-potsdam.de}\\
\\
Martin Kutrib and Andreas Malcher
\institute{Institut f\"ur Informatik, Universit\"at Giessen,\\
  Arndtstr.~2, 35392 Giessen, Germany}
  \email{$\{$kutrib,malcher$\}$@informatik.uni-giessen.de}
}

\begin{document}

\maketitle

\begin{abstract}
Systems of deterministic finite automata communicating by sending their 
states upon request are investigated, when the amount of communication is restricted.
The computational power and decidability properties are studied for the case 
of returning centralized systems, when the number of
necessary communications during the computations of the system
is bounded by a function depending on the length of the input.
It is proved that an infinite hierarchy of language families exists, depending
on the number of messages sent during their most economical recognitions. Moreover,
several properties are shown to be not semi-decidable for the systems under consideration.
\end{abstract}

\section{Introduction}\label{sec:intro}

Communication is one of the most fundamental concepts in computer science:
objects of object-oriented programs, roles or pools in business processes,
concurrent processes in computer networks or in information or operating systems 
are examples of communicating agents.

Parallel communicating finite automata systems (PCFA) have been introduced 
in~\cite{Martinvide:2002:pfascs} as a simple automaton model of parallel processes 
and cooperating systems, 
see also~\cite{bordihn:2011:uhrpcfa,bordihn:2012:ccpcfa,Choudhary:2007:rnrpcnfae}. 
A PCFA consists of several finite automata,
the components of the system, that process a joint input string 
independently of each other. However, their transitions are synchronized according 
to a global clock. The cooperation of the components is enabled by 
communication steps in which components can request the state reached 
by another component.
The system can work in returning or non-returning mode. In the former case
each automaton which sends its current state is set back to its initial state after this
communication step. In the latter case the state of the sending automaton is not
changed. 
Recently, these communication protocols have been refined in~\cite{vollweiler:2013:aspcfa}
and further investigated for the case of parallel
communicating systems of pushdown automata~\cite{otto:2013}. 
There, the communication process 
is performed in an asynchronous manner, reflecting the technical features of
many real communication processes. In the sequel of this paper and as a first step towards
an investigation of the influence of restricted communication to parallel communicating 
systems of automata, we stick with the simpler model having synchronized communication steps.

In a PCFA, one also distinguishes between centralized systems where only one designated 
automaton, called master, can request information from other automata, and
non-centralized systems where every automaton is allowed to request information
from others. Taking the distinction between returning and non-returning systems into account, 
we are led to four different working modes. Moreover, one distinguishes between 
deterministic and nondeterministic PCFA. The system is deterministic, 
if all its components are deterministic finite automata.

It is known from~\cite{bordihn:2012:ccpcfa,Choudhary:2007:rnrpcnfae,Martinvide:2002:pfascs}
that deterministic (nondeterministic) non-centralized PCFA are equally powerful 
as deterministic (nondeterministic) one-way multi-head finite automata~\cite{ibarra:1973:twmha}, 
both in returning and non-returning working modes. Moreover, it is proved in~\cite{bordihn:2012:ccpcfa} 
that nondeterminism is strictly more powerful than determinism 
for all the four working modes, and that deterministic centralized returning 
systems are not weaker than deterministic centralized non-returning ones.

All variants of PCFA accept non-regular languages due to the feature that
communication between the components of the system is allowed. 
Thus it is of interest to measure the amount of communication needed for 
accepting those languages. Mitrana proposed in~\cite{mitrana:2000:odcipcfas} 
a dynamical measure of descriptional complexity as follows:
The degree of communication of a PCFA for a given word is the minimal number of 
communications necessary to recognize the word. Then, the degree of communication 
of a PCFA is the supremum of the degrees of communication taken over all words 
recognized by the system, while the degree of communication of a language 
(with respect to a PCFA of type~$X$) is the infimum of the degrees of communication
taken over all PCFA of type~$X$ that accept the language. 
Mitrana proved that this measure cannot be algorithmically computed for languages 
accepted by nondeterministic centralized or non-centralized non-returning PCFA. 
The computability status of the degree of communication for the other types 
of PCFA languages as well as for all types of PCFA is stated as open question 
in~\cite{mitrana:2000:odcipcfas}. 

In this paper, we study PCFA where the degree of communication is bounded by a function
in the length of the input word. We restrict ourselves to one of the simplest types 
of PCFA, namely to deterministic centralized returning systems of finite automata.
In the next section, the basic definitions and two examples of languages accepted by
communication bounded PCFA are presented. In Section~\ref{sec:cap}, we show that bounding the 
degree of communication by logarithmic, square root or linear functions leads to
three different families of languages. For the strictness results, we use similar 
witness languages and a proof technique based on Kolmogorov complexity as 
in~\cite{kutrib:2010:tpwcc},
where the second and the third author investigated the computational power of 
two-party Watson-Crick systems, that is, synchronous systems consisting of two finite 
automata running in opposite directions on a shared read-only input and
communicating by broadcasting messages. 

In Section~\ref{sec:decidability}, non-semi-decidability results are proved for 
deterministic returning centralized PCFA and their languages, thus partially answering 
questions listed as open in~\cite{mitrana:2000:odcipcfas}. Similarly to~\cite{bordihn:2011:uhrpcfa}
the proofs rely on properties of one-way cellular automata and their valid computations. 
Finally, Section~\ref{sec:hierarchy} refines the three-level hierarchy from Section~\ref{sec:cap}
to an infinite hierarchy.

\section{Preliminaries and Definitions}\label{sec:def}

We write $\Sigma^*$ for the set of all words over the finite 
alphabet $\Sigma$, and $\mathbb{N}$ for the set $\{0,1,2,\dots\}$ 
of non-negative integers.
The \emph{empty word} is denoted by
$\lambda$.
For the \emph{length} of $w$ we
write $|w|$.
We use $\subseteq$ for \emph{inclusions} and $\subset$ for 
\emph{strict inclusions}.

Next we turn to the definition of 
parallel communicating finite automata systems. The nondeterministic
model has been introduced in~\cite{Martinvide:2002:pfascs}.
Following~\cite{bordihn:2011:uhrpcfa}, the formal definition is as follows.

A \emph{deterministic parallel communicating finite automata system of
degree~$k$} $(\dpcfa(k))$ is a construct 
$A=\langle \Sigma, A_1, A_2, \dots, A_k, Q, \rightend\rangle$, 
where 
\begin{enumerate}
\item
$\Sigma$ is the set of \emph{input symbols}, 
\item
each $A_i=\langle S_i, \Sigma,
\delta_i, s_{0,i}, F_i\rangle$, $1\leq i\leq k$, is a 
\emph{deterministic finite automaton} with finite state set~$S_i$, 
\emph{partial} transition function 
\mbox{$\delta_i: S_i \times (\Sigma \cup \{\lambda,\rightend\}) \to S_i$}
(requiring that $\delta_i(s,a)$ is undefined for all $a \in \Sigma \cup \{\rightend\}$, if 
$\delta_i(s,\lambda)$ is defined),
initial state $s_{0,i}\in S_i$, and set of
accepting states $F_i\subseteq S_i$, 
\item
 $Q=\{q_1, q_2, \dots, q_k\}\subseteq \bigcup_{1 \leq i \leq k} S_i$
 is the set of \emph{query states}, and
\item
$\rightend\notin \Sigma$ is the \emph{end-of-input symbol}.
\end{enumerate}

The single automata are called \emph{components} of the
system $A$.
A \emph{configuration}  $(s_1, x_1, s_2, x_2, \dots, s_k, x_k)$ of 
$A$ represents the current states $s_i$ 
as well as the still unread parts $x_i$ of the tape inscription of all 
components $1\leq i\leq k$. System $A$ starts with all of 
its components scanning the first square of the tape in their initial states. 
For input word $w\in \Sigma^*$, the initial configuration is
$(s_{0,1}, w\rightend, s_{0,2}, w\rightend, \dots, s_{0,k}, w\rightend)$.

Basically, a computation of $A$ is a sequence of configurations
beginning with an initial configuration and ending with a halting
configuration, when no successor configuration exists. 
Each step can consist of two phases. 
In a first phase, all components are in non-query states and
perform an ordinary (non-communicating) step independently.
The second phase is the communication phase during which components 
in query states receive the requested states as long as the sender 
is not in a query state itself. 
That is, if a component $A_i$ is in query
state $q_j$, then $A_i$ is set to the current state of component $A_j$.
This process is repeated until all requests 
are resolved, if possible. If the requests are cyclic, no successor configuration exists.
For the first phase, we define the successor configuration relation 
$\vdash$ by 
$
(s_1, a_1y_1, s_2, a_2y_2, \dots, s_k, a_ky_k) \vdash 
(p_1, z_1, p_2, z_2, \dots, p_k, z_k),
$
if $Q\cap \{s_1,s_2,\dots,s_k\} = \emptyset$, 
$a_i\in \Sigma\cup\{\lambda,\rightend\}$, $p_i\in \delta_i(s_i,a_i)$, 
and $z_i=\rightend$ for $a_i =\rightend$ and $z_i=y_i$ otherwise,
$1\leq i \leq k$.
For non-returning communication in the second phase, we set 
$
(s_1, x_1, s_2, x_2, \dots, s_k, x_k) \vdash 
(p_1, x_1, p_2, x_2, \dots, p_k, x_k), 
$
if, for all $1\leq i\leq k$ such that $s_i = q_j$ and $s_j\notin Q$, we
have $p_i=s_j$, and $p_r=s_r$ for all the other $r$, $1\leq r\leq k$.
Alternatively, for returning communication in the second phase, 
we set 
$
(s_1, x_1, s_2, x_2, \dots, s_k, x_k) \vdash 
(p_1, x_1, p_2, x_2, \dots, p_k, x_k), 
$
if, for all $1\leq i\leq k$ such that $s_i = q_j$ and $s_j\notin Q$, we
have $p_i=s_j$, $p_j=s_{0,j}$, and $p_r=s_r$ for all the other~$r$, $1\leq r\leq k$.

A computation \emph{halts} when the successor configuration is 
not defined for the current situation. In particular, this may happen
when cyclic communication requests appear, or when the transition
function of one component is not defined. 
The language $L(A)$ accepted by a $\dpcfa(k)$ $A$
is precisely the set of words~$w$ such that there is some computation beginning with 
$w\rightend$ on the input tape and halting with at least one component
having an undefined transition function and 
being in an accepting state. Let $\vdash^*$ denote the reflexive and
transitive closure of the successor configuration relation~$\vdash$ and define
$L(A)$ as
\begin{multline*}
\{\,w\in \Sigma^* \mid 
(s_{0,1}, w\rightend, s_{0,2}, w\rightend, \dots, s_{0,k}, w\rightend)
\vdash^*
(p_1, a_1y_1, p_2, a_2y_2, \dots, p_k, a_ky_k),\\
\mbox{ such that }p_i\in F_i \mbox{ and }\delta_i(p_i,a_i)\mbox{ as well as }\delta_i(p_i,\lambda) \mbox{ are undefined
for some }1\leq i\leq k\,\}.
\end{multline*}

Whenever the degree is missing in the notation $\dpcfa(k)$, 
we mean systems of arbitrary degree.
The absence or presence of an R in the type of the system denotes
whether it works in \emph{non-returning} communication, that is, the sender 
remains in its current state, or \emph{returning} communication, that is, 
the sender is reset to its initial state. If there 
is just one component, say $A_1$, that is allowed to query for states, that
is, $S_i\cap Q=\emptyset$, for $2\leq i\leq k$, then the system is said to be
\emph{centralized}. In this case, we refer to~$A_1$ as the 
\emph{master component} and add a~C to the notation of the type of the system.
The \emph{family of languages accepted} by devices of type~$X$ with arbitrary degree
(with degree $k$) is denoted by~$\lfam(X)$ ($\lfam(X(k))$).

In the following, we study the impact of communication in 
$\pcfa$. The communication is measured by the \emph{total number 
of queries sent during a computation}. That is, we count the 
number of time steps at which a component enters a query
state and consider the sum of these numbers for all components.
Let $f:\mathbb{N}\to\mathbb{N}$ be a mapping. 
If all $w\in L(A)$ are accepted with computations where 
the total number of queries sent is bounded by $f(|w|)$,
then $A$ is said to be \emph{communication bounded by~$f$.}

We denote the class of devices of type~$X$
(with degree $k$) that are communication bounded by some
function $f$ by $f\textrm{-X}$ ($f\textrm{-X}(k)$).

In order to clarify the notation we give two examples.
Whenever we refer to a time $t$ of a computation of a $\dpcfa$, then the 
configuration reached after exactly $t$ computation steps is considered.

\begin{example}\label{exa:expo}
The language 
$L_{expo}=\{\,\dollar a^{2^{0}}b a^{2^{1}} b \cdots ba^{2^{m}} \kand \mid m \ge 1\,\}$ 
belongs to $\lfam(\fdrcpcfa{f}(2))$ with $f \in O(\log(n))$.
Roughly, the idea of the construction is that the lengths of adjacent $a$-blocks 
(separated by a $b$) are compared. To this end, the master reads
the left block with half speed, that is, moving one symbol to the
right in every other time step, while the non-master component reads the right
block with full speed, that is, moving one symbol to the
right in every time step. If the master reaches a $b$, it queries the
non-master whether it has also reached a $b$. If this is true, the comparison
of the next two $a$-blocks is started. The input is accepted if the
master obtains the symbol $\kand$ from the non-master component and the 
remaining input is in $a^+\kand\rightend$.

Formally, we define $A=\langle \{a,b,\dollar,\kand\}, A_1, A_2, \{q_2\}, \rightend \rangle$
to be a $\drcpcfa(2)$ with master component
$A_1=\langle \{s_{0,1},s_{1,1},s_{2,1},s_{3,1},s_{4,1},s_{5,1},s_b,s_{\kand},q_2,accept\}, \{a,b,\dollar,\kand\},
\delta_1, s_{0,1}, \{accept\}\rangle$, second component
$A_2=\langle \{s_{0,2},s_{1,2},s_{2,2},s_{3,2},s_b,s_{\kand},s_{\rightend}\}, \{a,b,\dollar,\kand\},
\delta_2, s_{0,2}, \emptyset\rangle$, and 
transition functions $\delta_1$ and $\delta_2$ as follows.\\[3mm]
\noindent
\underline{The non-master component $A_2$:}
\begin{multicols}{3}
\begin{enumerate}
\item $\delta_2(s_{0,2},\dollar) = s_{1,2}$
\item $\delta_2(s_{1,2},a) = s_{2,2}$
\item $\delta_2(s_{2,2},b) = s_{3,2}$
\item $\delta_2(s_{3,2},a) = s_{3,2}$
\item $\delta_2(s_{3,2},b) = s_b$
\item $\delta_2(s_{3,2},\kand) = s_{\kand}$
\item $\delta_2(s_{0,2},a) = s_{3,2}$
\item $\delta_2(s_{0,2},\rightend) = s_{\rightend}$
\item $\delta_2(s_{\rightend},\lambda) = s_{\rightend}$
\end{enumerate}
\end{multicols}

\noindent
The component reads the input prefix $\dollar ab$ in the first three time steps
(rules 1,2,3). Subsequently, it reads an $a$-block in state $s_{3,2}$ (rule 4).
Whenever it moves on a symbol $b$ it changes into state $s_b$ (rule 5).
So, it enters state $s_b$ at time step 3 plus the length of the second $a$-block plus 1.
The component halts in state~$s_b$ unless it is reset to its initial state
by a query. In this case it reads the current $a$-block and the next $b$ 
and enters state $s_b$ again after a number of time steps that is the length
of the $a$-block plus one (rules 7,4,5). Rule 6 is used when $\kand$ appears in the input
instead of $b$. After being reset into the initial state on the endmarker,
the component enters state $s_{\rightend}$ and loops with $\lambda$-moves.\\[3mm]
\noindent
\underline{The master component $A_1$:}
\begin{multicols}{3}
\begin{enumerate}
\item $\delta_1(s_{0,1},\dollar) = s_{1,1}$
\item $\delta_1(s_{1,1},\lambda) = s_{2,1}$
\item $\delta_1(s_{2,1},\lambda) = s_{3,1}$
\item $\delta_1(s_{3,1},a) = s_{4,1}$
\item $\delta_1(s_{4,1},\lambda) = s_{3,1}$
\item $\delta_1(s_{3,1},b) = q_2$
\item $\delta_1(s_b,a) = s_{4,1}$
\item $\delta_1(s_{\kand},a) = s_{\kand}$
\item $\delta_1(s_{\kand},\kand) = s_{5,1}$
\item $\delta_1(s_{5,1},\rightend) = accept$
\end{enumerate}
\end{multicols}

\noindent
The master reads the input prefix $\dollar ab$ in the first six time steps
and enters the query state $q_2$ (rules 1--6). Exactly at that time the non-master 
component enters state $s_b$. Being in state $s_b$ received the master reads the current 
$a$-block and the next $b$ and enters state $q_2$ again after a number of time 
steps that is two times the length of the $a$-block plus one (rules 7,4,5,6).
Exactly at this time the non-master component enters state $s_b$ again provided
that the $a$-block read by the non-master component is twice as long as the $a$-block
read by the master.
When the master receives state $s_{\kand}$ instead of $s_b$, it reads the
remaining suffix \mbox{(rules 8,9),} enters the accepting state on the endmarker
(rule 10) and halts.

Finally, the length of a word $w\in L_{expo}$ is
$|w|=m+2+\sum_{i=0}^{m}2^i=2^{m+1}+m+1$,
for some $m\geq 1$. In its accepting computation, a communication
takes place for every symbol $b$ and the endmarker. So
there are $m+1$ communications which is of order~$O(\log(|w|))$.
\qed
\end{example}

The construction of the next example is similar to the one given in Example~\ref{exa:expo}.

\begin{example}\label{exa:poly}
The language
$L_{poly}=\{\,\dollar ab a^{3} b a^{5} b\cdots b a^{2m+1} \kand \mid \mbox{$m \ge 0\,$}\,\}$ 
belongs to $\lfam(\fdrcpcfa{f}(2))$  with $f \in O(\sqrt{n})$.
\qed
\end{example}

\section{Computational Capacity}\label{sec:cap}

In this section we consider aspects of the computational capacity of
$\fdrcpcfa{f}(k)$. Examples~\ref{exa:expo} and~\ref{exa:poly} already
revealed that there are non-semilinear languages accepted by systems
with two components and sublinear communication. The next simple result
is nevertheless important for the size of representations that will be
used in connection with Kolmogorov arguments to separate language classes.

\begin{lemma}\label{lem:linear:time}
Let $k\geq 1$ and $A$ be a $\drcpcfa(k)$ with $S_1, S_2, \ldots, S_k$ being 
the state sets of the single components. 
If $w \in L(A)$, then $w$ is
accepted after at most $|S_1|\cdot|S_2|\cdots|S_k|\cdot(|w|+1)$ time steps, that is,
in linear time.
\end{lemma}

\begin{proof}
During a computation some component $A_i$ may be in $|S_i|$ different states.
So after $|S_1| \cdot |S_2| \cdots |S_k|$ time steps 
the whole system runs through a loop if none of the components moves.
Therefore, as long as no halting configuration is
reached, at least one component must move 
after at most $|S_1|\cdot|S_2|\cdots|S_k|$ time steps.
\end{proof}

The language of the next lemma combines the well-known non-context-free copy
language with $L_{expo}$ from above. It plays a crucial role in later proofs.

\begin{lemma}\label{lem:expoww}
The language 
\[
L_{expo,wbw}=\{\, \dollar w_1 w_2 \cdots w_m b a^{2^{0}}w_1w_1 a^{2^{1}} w_2w_2 \cdots a^{2^{m-1}} w_mw_m  \kand
\mid m \ge 1, w_i \in \{0,1\}, 1 \le i \le m\,\}
\]
belongs to $\lfam(\fdrcpcfa{O(\log(n))}(3))$.
\end{lemma}

\begin{proof}
A formal construction of a $\fdrcpcfa{O(\log(n))}(3)$ accepting $L_{expo,wbw}$
is given through the transition functions below, where 
$s_{0,i}$ is the initial state of component $A_i$, $1\leq i\leq 3$, 
the sole accepting state is $accept$, and $\sigma \in \{0,1\}$.

The second non-master component $A_3$
initially passes over the $\dollar$ and, then, it reads a symbol, remembers it
in its state, and loops without moving (rules 1,2,3,8,9). Whenever the 
component is reset into its initial state after a query, it reads the next 
symbol, remembers it, and loops without
moving (rules 4--11). This component is used by the master 
to match the $w_i$ from the prefix with the $w_i$ from the suffix.\\[3mm]
\noindent
\underline{The non-master component $A_3$:}
\begin{multicols}{3}
\begin{enumerate}
\item $\delta_3(s_{0,3},\dollar) = s_{1,3}$
\item $\delta_3(s_{1,3},0) = s_{0}$
\item $\delta_3(s_{1,3},1) = s_{1}$
\item $\delta_3(s_{0,3},0) = s_{0}$
\item $\delta_3(s_{0,3},1) = s_{1}$
\item $\delta_3(s_{0,3},b) = s_{b}$
\item $\delta_3(s_{0,3},a) = s_{a}$
\item $\delta_3(s_{0},\lambda) = s_{0}$
\item $\delta_3(s_{1},\lambda) = s_{1}$
\item $\delta_3(s_{b},\lambda) = s_{b}$
\item $\delta_3(s_{a},\lambda) = s_{a}$
\item[]
\end{enumerate}
\end{multicols}

The first non-master component $A_2$ initially passes over the prefix 
$\dollar w_1 w_2 \cdots w_m$ (rules 1,2), the
$b$ (rule 3), and the adjacent infix $aw_1w_1aaw_2w_2$ (rules 4--13). 
On its way it checks whether the neighboring symbols $w_i$ are in fact the same 
%and remembers it in its state 
(rules 5--8 and 10--13). If the second check is successful the
component enters state $s_{ww}$. Exactly at that time it has to be queried by
the master, otherwise it blocks the computation. Subsequently, it repeatedly 
continues to read the input, where each
occurrence of neighboring symbols $w_i$ are checked for equality (rules 14 and
9--13), which is indicated by entering state $s_{ww}$ again.
This component is used to verify that all neighboring symbols $w_i$ in the
suffix are equal and, by the master, to check the lengths of the $a$-blocks in
the same way as in Example~\ref{exa:expo}. Note that the component
is at time $m+9$ on the first symbol after $w_2w_2$. After being
reset to its initial state, it takes a number of time steps equal to
the length of the next $a$-block plus 2 to get on the first symbol 
after the next $w_iw_i$.\\[3mm]
\noindent
\underline{The non-master component $A_2$:}
\begin{multicols}{3}
\begin{enumerate}
\item $\delta_2(s_{0,2},\dollar) = s_{1,2}$
\item $\delta_2(s_{1,2},\sigma) = s_{1,2}$
\item $\delta_2(s_{1,2},b) = s_{2,2}$
\item $\delta_2(s_{2,2},a) = s_{3,2}$
\item $\delta_2(s_{3,2},0) = s^0_{4,2}$
\item $\delta_2(s_{3,2},1) = s^1_{4,2}$
\item $\delta_2(s^0_{4,2},0) = s_{5,2}$
\item $\delta_2(s^1_{4,2},1) = s_{5,2}$
\item $\delta_2(s_{5,2},a) = s_{5,2}$
\item $\delta_2(s_{5,2},0) = s^0_{6,2}$
\item $\delta_2(s_{5,2},1) = s^1_{6,2}$
\item $\delta_2(s^0_{6,2},0) = s_{ww}$
\item $\delta_2(s^1_{6,2},1) = s_{ww}$
\item $\delta_2(s_{0,2},a) =  s_{5,2}$
\item $\delta_2(s_{0,2},\kand) = s_{\kand}$
\item $\delta_2(s_{\kand},\lambda) = s_{\kand}$
\item $\delta_2(s_{0,2},\rightend) = s_{\rightend}$
\item $\delta_2(s_{\rightend},\lambda) = s_{\rightend}$
\end{enumerate}
\end{multicols}

The master component $A_1$ initially passes over the prefix 
$\dollar w_1 w_2 \cdots w_m$ (rules 1,2), the
$b$ (rule 3), and the first $a$ (rules 4--8). Then it reads the first of
two adjacent symbols $w_i$ and enters the query state $q_3$ (rule 9) (the equality 
of the symbols $w_i$ has already been checked by component $A_2$).
From component $A_3$ it receives the information about the matching symbol
$w_i$ from the prefix. If this symbol is the same as the next input symbol, then
the computation continues (rules 10,11) by entering query state $q_2$.
Note that this happens exactly at time step $m+9$. If the master receives
state $s_{ww}$ the length of the first two $a$-blocks are verified.
Now the master repeatedly continues to read the input (rule 12,7,8),
where on each occurrence of neighboring symbols $w_i$ the equality with
the corresponding symbol in the prefix is checked by querying component~$A_3$
and the lengths of the $a$-blocks are compared by querying component $A_2$.
After querying component $A_2$, it takes a number of time steps equal to
the length of the adjacent $a$-block (processed by component~$A_2$) plus 2 to get into state $q_2$ again.
Finally, when the master component has checked the last symbol $w_m$ and gets the information that
$A_2$ has read symbol $\kand$, it queries component $A_3$ (rule 13). If it 
receives a $b$, the input is accepted (rule 14). In all other cases it is rejected.\\[3mm]
\noindent
\underline{The master component $A_1$:}
\begin{multicols}{3}
\begin{enumerate}
\item $\delta_1(s_{0,1},\dollar) = s_{1,1}$
\item $\delta_1(s_{1,1},\sigma) = s_{1,1}$
\item $\delta_1(s_{1,1},b) = s_{2,1}$
\item $\delta_1(s_{2,1},\lambda) = s_{3,1}$
\item $\delta_1(s_{3,1},\lambda) = s_{4,1}$
\item $\delta_1(s_{4,1},\lambda) = s_{5,1}$
\item $\delta_1(s_{5,1},a) = s_{6,1}$
\item $\delta_1(s_{6,1},\lambda) = s_{5,1}$
\item $\delta_1(s_{5,1},\sigma) = q_{3}$
\item $\delta_1(s_0,0) = q_2$
\item $\delta_1(s_1,1) = q_2$
\item $\delta_1(s_{ww},a) = s_{6,1}$
\item $\delta_1(s_{\kand},\kand) = q_{3}$
\item $\delta_1(s_{b},\rightend) = accept$
\item[]
\end{enumerate}
\end{multicols}

The length of a word $w\in L_{expo,wbw}$ is
$|w|=3m+3+\sum_{i=0}^{m-1}2^i=2^{m}+3m+2$,
for some $m\geq 1$. In its accepting computation, two communications
take place for every $w_iw_i$ and one more communication on the endmarker. So
there are $2m+1$ communications which is of order~$O(\log(|w|))$.
\end{proof}

For the proof of the following theorem we use an incompressibility
argument. General information on Kolmogorov complexity and the
incompressibility method can be found in~\cite{li:1993:itkca:book}.
Let $w\in \{0,1\}^+$ be an arbitrary binary string. The Kolmogorov 
complexity $C(w)$ of $w$ is defined to be the minimal size of a 
program describing~$w$. The following key argument for the incompressibility
method is well known. There are binary strings~$w$ 
of \emph{any} length so that $|w| \le C(w)$. 

\begin{lemma}\label{lem:wbw}
The language 
$L_{wbw}=\{\,w_1 w_2 \cdots w_m b w_1 w_2 \cdots w_m
\mid m \ge 1, w_i \in \{0,1\}, 1 \le i \le m\,\}$ 
is accepted by some $\fdrcpcfa{O(n)}(2)$ but, for any $k\geq 1$,
does not belong to $\lfam(\fdrcpcfa{f}(k))$ 
if $f \in \frac{n}{\omega(\log(n))}$.
\end{lemma}

\begin{proof}
First, we sketch the construction of a $\fdrcpcfa{O(n)}(2)$ accepting $L_{wbw}$.
Initially, the master component proceeds to the center marker $b$,
while the non-master component reads the first input symbol~$w_1$ and 
remembers this information in its state. Next, the master queries
the non-master and matches the information received with the first symbol following $b$,
while the non-master reads the next input symbol and remembers it in its state.
Subsequently, this behavior is iterated, that is, the master queries the
non-master again and matches its next input symbol, while the non-master 
reads and remembers the next symbol. The input is accepted when the master 
receives a $b$ at the moment it reaches the right endmarker. Clearly, 
the number of communications on input length $n=2m+1$ is $m+1 \in O(n)$.

Second, we turn to show that $L_{wbw}\notin\lfam(\fdrcpcfa{f}(k))$ 
if $f \in \frac{n}{\omega(\log(n))}$. In contrast to the assertion,
we assume that $L_{wbw}$ is accepted by some
$\fdrcpcfa{f}(k)$ $A=\langle \Sigma, A_1, A_2, \ldots, A_k, Q, \rightend\rangle$ 
with \mbox{$f(n) \in \frac{n}{\omega(\log(n))}$.}
Let $z=wbw$, for some $w \in \{0,1\}^+$, and 
$K_0 \vdash \cdots \vdash K_{acc}$
be the accepting computation on input $z$, where $K_0$ is the
initial configuration and $K_{acc}$ is an accepting configuration.

Next, we consider snapshots of configurations at every time step
at which the master component queries some other component
or at which a component enters the middle marker $b$.
For every such configuration, we take the
time step~$t_i$, the current states $s_1^{(i)}, s_2^{(i)}, \ldots, s_k^{(i)}$,
and the positions $p_1^{(i)}, p_2^{(i)}, \ldots, p_k^{(i)}$ of the components.
Thus, the $i$th snapshot is represented by the tuple
$(t_i,s_1^{(i)},p_1^{(i)},s_2^{(i)},p_2^{(i)},\ldots,s_k^{(i)},p_k^{(i)})$.
Since there are altogether at most \mbox{$f(2|w|+1)$} communications, 
the list of snapshots $\Lambda$ contains at most $f(2|w|+1)+k$ entries. 

We claim that each snapshot can be represented by at most $O(\log(|w|))$ bits.
Due to Lemma~\ref{lem:linear:time} acceptance is in linear time and,
therefore, each time step can be represented by at most $O(\log(|w|))$ bits.
Each position of a component can also be represented by at most $O(\log(|w|))$ bits. Finally, 
each state can be represented by a constant number of
bits. Altogether, each snapshot can be represented by $O(\log(|w|))$ bits. 
So, the list $\Lambda$ can be represented by
$(f(2|w|+1)+k) \cdot O(\log(|w|))= \frac{|w|}{\omega(\log(|w|))}\cdot 
O(\log(|w|))=o(|w|)$ bits.

Now we show that the list $\Lambda$ of snapshots together with 
a snapshot of~$K_{acc}$ and
the knowledge of $A$ and $|w|$ is sufficient to reconstruct $w$.
The reconstruction is implemented by the following algorithm $P$.
First, $P$ sequentially simulates $A$ on all $2^{|w|}$ inputs $xbx$ 
where $|x|=|w|$.
Additionally, it is checked whether the computation simulated has the
same snapshots as in the list $\Lambda$ and the accepting configuration. 
In this way, the string $w$ can be identified. 
We have to show that there is no other
string $w' \neq w$ which can be identified in this way as well. 
Let us assume that such a $w'$ exists. Then all snapshots of 
accepting computations on input $wbw$ and $w'bw'$ are identical. This means
that both computations end at the same time step and all
components are in the same state and position. Additionally, in both
computations communications take place at the same time steps, all
components are in the same state and position at that moment.
Moreover, the right half of the respective words is entered in the same states
and in the same time steps on both input words $wbw$ and $wbw'$.
So, both computations are also accepting on input~$wbw'$
which is a contradiction.

Thus, $w$ can be reconstructed given the above program $P$, the list of
snapshots~$\Lambda$, the snapshot of the accepting configuration, 
$A$, and $|w|$. Since the sizes of $P$ and $A$
are bounded by a constant, the size of $\Lambda$ is bounded 
by~$o(|w|)$, and~$|w|$ as well as the size of the remaining snapshot
is bounded by $O(\log(|w|))$ each, we can reconstruct $w$ from a 
description of total size $o(|w|)$.
Hence, the Kolmogorov complexity $C(w)$, that is, the minimal size of a program
describing $w$ is bounded by the size of the above description, and
we obtain $C(w) \in o(|w|)$.
On the other hand, we know that there are binary strings $w$ of arbitrary length
such that $C(w) \ge |w|$. This is a contradiction for $w$ being long enough.
\end{proof}

The language of the next lemma 
is used in later proofs.

\begin{lemma}\label{lem:polyww}
The language 
\[
L_{poly,wbw}=\{\, \dollar w_1 w_2 \cdots w_m b a^{1} w_1w_1 a^{3} w_2w_2 a^{5} w_3w_3 \cdots a^{2m-1} w_mw_m  \kand
\mid m \ge 1, w_i \in \{0,1\}, 1 \le i \le m\,\}
\]
is accepted by some $\fdrcpcfa{O(\sqrt{n})}(3)$ but, for any $k\geq 1$,
does not belong to 
$\lfam(\fdrcpcfa{f}(k))$ if $f \in O(\log(n))$.
\end{lemma}

\begin{proof}
Using the construction idea of Lemma~\ref{lem:expoww}, %and~\ref{lem:wbw} 
one shows $L_{poly,wbw}\in \lfam(\fdrcpcfa{O(\sqrt{n})}(3))$.

The claimed non-containment is shown similarly to Lemma~\ref{lem:wbw}: 
in contrast to the assertion, we assume that $L_{poly,wbw}$ is accepted by some
$\fdrcpcfa{f}(k)$ \mbox{$A=\langle \Sigma, A_1, A_2, \ldots, A_k, Q, \rightend\rangle$} 
with $f(n)\in O(\log(n))$. Let
\[
z=\dollar w_1 w_2 \cdots w_m b a^{1} w_1w_1 a^{3} w_2w_2 a^{5} w_3w_3 \cdots
a^{2m-1} w_mw_m  \kand \in L_{poly,wbw},
\] 
where
$w=w_1w_2\cdots w_m$, and  $K_0 \vdash \cdots \vdash K_{acc}$
be the accepting computation on input $z$, where $K_0$ is the
initial configuration and $K_{acc}$ is an accepting configuration.

We use again an incompressibility argument and write down
the list of snapshots of configurations in which communication 
takes place and the accepting configuration~$K_{acc}$, and 
descriptions of $A$ and $|w|$. 
Similar to the proof of Lemma~\ref{lem:wbw}, a 
program~$P$ can be described which reconstructs $w$ uniquely from 
the information given.

Next, we determine the size of such a description. Program~$P$ and
the system~$A$ can be represented by a constant number of bits. The length
$|w|$ can be described by $\log(|w|)\in O(\log(m))$ bits. 
Since $|z|=3m+3+\sum_{i=1}^{m} 2i-1=3m+3+m^2$ and acceptance
is in linear time (Lemma~\ref{lem:linear:time}), each time step
can be represented by $O(\log(|z|))= O(\log(m^{2}))$ bits.
Moreover, the $k$ states can be described by $O(1)$ bits, 
and the $k$ positions by $k\cdot\log(|z|)=k\cdot\log(m^{2}+3m+3)
\in O(\log(m))$ bits. So, altogether one snapshot can be represented by
$O(\log(m))$ bits. 
Since at most $f(|z|)\in O(\log(|z|))= O(\log(m))$ snapshots have to be listed,
the list of all snapshots can be described by $O((\log(m))^{2})$ bits.
Therefore, the total size of a description of $w$ is bounded by $O((\log(m))^{2})$ 
as well.
Thus, the Kolmogorov complexity $C(w)$ of $w$ is bounded by $O((\log(m))^{2})$.
On the other hand, there are binary strings $w$ of arbitrary length
such that $C(w) \ge |w|=m$. This is a contradiction for $w$ being long enough.
\end{proof}

The previous theorems showed that there are proper inclusions
\[
\lfam(\fdrcpcfa{O(\log(n))}(k)) \subset \lfam(\fdrcpcfa{O(\sqrt{n})}(k))
\]
for every $k \ge 3$, and
\[
\lfam(\fdrcpcfa{O(\sqrt{n})}(k)) \subset \lfam(\fdrcpcfa{O(n)}(k))
\]
for every $k \ge 2$.

Later, we will prove an infinite hierarchy in between
the classes $\lfam(\fdrcpcfa{O(\log(n))}(k))$ and $\lfam(\fdrcpcfa{O(\sqrt{n}))}(k)$,
for every $k \ge 4$.

\section{Decidability and Undecidability Results}\label{sec:decidability}

\subsection{Undecidability of Emptiness and Classical Questions}

First, we show undecidability of the classical questions for models with
a logarithmic amount of communication. To this end, we adapt the construction
given in~\cite{bordihn:2011:uhrpcfa} which is based on the valid computations of
\emph{one-way cellular automata} ($\oca$), a parallel computational model 
(see, for example, \cite{kutrib:2008:ca-cpv,kutrib:2009:calt}).
More precisely, the undecidability is shown by 
reduction of the corresponding problems for $\oca$
which are known not even to be semi-decidable~\cite{Malcher:2002:dccadq}.
To this end, histories of $\oca$ computations are encoded 
in single words that are called \emph{valid computations} 
(cf., for example,~\cite{Hopcroft:1979:itatlc:book}).

A one-way cellular automaton is a linear array of
identical deterministic finite automata, sometimes called cells.
Except for the leftmost cell each one is connected to its nearest neighbor
to the left. The state transition depends on the current state of a cell itself
and the current state of its neighbor, where the leftmost cell 
receives information associated with a boundary symbol 
on its free input line. The state changes
take place simultaneously at discrete time steps.
The input mode for cellular automata is called parallel. 
One can suppose that all cells fetch their input symbol 
during a pre-initial step.

More formally, an $\oca$ is a system $M=\langle S,\texttt{\#},
T,\delta,F \rangle$, where 
$S$ is the nonempty, finite set of cell states, 
$\texttt{\#} \notin S$ is the boundary state,
$T \subseteq S$ is the input alphabet, 
$F \subseteq S$ is the set of accepting cell states, and
$\delta : (S \cup \{\texttt{\#}\}) \times S \to S$ is the local 
transition function.

A configuration of an $\oca$ at some time step $t \ge 0$ is a
description of its global state, which is formally a mapping
$c_t:\{1,2,\dots,n\} \to S$, for $n\geq 1$.
The initial configuration at time~$0$ on input \mbox{$w=x_1x_2\ldots x_n$}
is defined by $c_{0,w}(i)=x_i$, $1 \le i \le n$. 
Let $c_t$, $t\geq 0$, be a configuration with $n\geq 2$, then its 
successor $c_{t+1}$ is defined as follows:
$c_{t+1}(1)=\delta(\texttt{\#},c_t(1))$ and 
$c_{t+1}(i)=\delta(c_t(i-1),c_t(i))$, $2\leq i\leq n$.

An input is accepted if at some time step
during its computation the rightmost cell enters an
accepting state.    
Without loss of generality and for technical reasons, one can assume that any
accepting computation has at least three steps.    

Now we turn to the valid computations of an $\oca$
$M=\langle S,\texttt{\#},T,\delta,F \rangle$.
The computation of a successor configuration $c_{t+1}$ of a given
configuration $c_{t}$
is written down in a sequential way as follows.
Assume~$c_{t+1}$ is computed cell by cell from left to right.
That is, we are concerned with subconfigurations of the form 
$c_{t+1}(1) \cdots c_{t+1}(i) c_{t}(i+1) \cdots c_t(n)$,
where $n$ is the length of the input. 
For technical reasons,
in $c_{t+1}(i)$ we have to store both the successor state, which is entered in
time step~$t+1$ by cell~$i$, and its former state.
In this way, the 
computation of the successor configuration of $M$ can be written as a 
sequence of $n$ subconfigurations, and 
configuration~$c_{t+1}$ can be represented by 
$w^{(t+1)}=w_1^{(t+1)} \cdots w_n^{(t+1)}$ such that $w_i^{(t+1)} \in \texttt{\#}
S^*(S\times S)S^*$, for $1 \le i \le n$, with \mbox{$w_i^{(t+1)}=\texttt{\#}c_{t+1}(1) 
\cdots c_{t+1}(i-1)(c_{t+1}(i),c_{t}(i)) c_{t}(i+1) \cdots c_t(n)$.}
The valid computations $\valc(M)$ are now defined to be 
the set of words of the form 
$w^{(0)} w^{(1)} \cdots w^{(m)}$, where $m\geq 3$,
$w^{(t)}\in (\texttt{\#}S^*(S\times S)S^*)^+$ are configurations of $M$, 
$1\le t\le m$, $w^{(0)}$ is an initial
configuration having the form $\texttt{\#}(T')^+$, where $T'$ 
is a primed copy of the input alphabet $T$ with $T'\cap S=\emptyset$, 
$w^{(m)}$ is an accepting configuration of the form 
$(\texttt{\#}S^*(S\times S)S^*)^*\texttt{\#}S^*(F \times S)$, 
and $w^{(t+1)}$ is the 
successor configuration of $w^{(t)}$, for $0 \leq t \leq m-1$.

For the constructions of $\drcpcfa$ accepting
the set $\valc(M)$, we provide an additional technical transformation of 
the input alphabet. Let $S'=S\cup T'$ and $A=\{\texttt{\#}\} \cup S' \cup S'^2$ 
be the alphabet over which
$\valc(M)$ is defined. We consider the mapping $f : A^+ \rightarrow 
(A \times A)^+$ which is defined for words of length at least two by 
$f(x_1 x_2 \cdots x_n)=[x_1,x_2][x_2,x_3] \cdots [x_{n-1},x_n]$. From now on we 
consider 
\mbox{$\valc(M) \subseteq (A \times A)^+$} to be the set of valid computations 
to which $f$ has been applied.
The set of \emph{invalid computations} $\invalc(M)$ 
is then the complement of $\valc(M)$ with respect to the alphabet~$A \times A$. 

The following example illustrates the definitions.

\begin{example}
We consider the following computation of an $\oca$ $M$ over the input 
alphabet $\{c,d\}$. The initial configuration is $c_0=(c,d,d)$. Let
the successor configurations be $c_1=(p_1,r_1,s_1)$, $c_2=(p_2,r_2,s_2)$, and 
$c_3=(p_3,r_3,s_3)$. Furthermore, let $s_3$ be an accepting state, that is, $cdd$ 
is an accepted input. These configurations are written down as sequences of
subconfigurations as follows.
\begin{eqnarray*}
w^{(0)} &=& \texttt{\#}c'd'd'\\
w^{(1)} &=& \texttt{\#}(p_1,c)dd\texttt{\#}p_1(r_1,d)d\texttt{\#}p_1r_1(s_1,d)\\
w^{(2)} &=& \texttt{\#}(p_2,p_1)r_1s_1\texttt{\#}p_2(r_2,r_1)s_1\texttt{\#}p_2r_2(s_2,s_1)\\
w^{(3)} &=& \texttt{\#}(p_3,p_2)r_2s_2\texttt{\#}p_3(r_3,r_2)s_2\texttt{\#}p_3r_3(s_3,s_2)
\end{eqnarray*}
Then, 
\begin{eqnarray*}
&&f(w^{(0)}w^{(1)}w^{(2)}w^{(3)})=[\texttt{\#},c'][c',d'][d',d'][d',\texttt{\#}]
[\texttt{\#},(p_1,c)][(p_1,c),d][d,d][d,\texttt{\#}]\\
&&[\texttt{\#},p_1][p_1,(r_1,d)][(r_1,d),d][d,\texttt{\#}][\texttt{\#},p_1][p_1,r_1][r_1,(s_1,d)][(s_1,d),\texttt{\#}]
[\texttt{\#},(p_2,p_1)]\\
&&[(p_2,p_1),r_1][r_1,s_1][s_1,\texttt{\#}][\texttt{\#},p_2][p_2,(r_2,r_1)][(r_2,r_1),s_1]
[s_1,\texttt{\#}][\texttt{\#},p_2][p_2,r_2]\\
&&[r_2,(s_2,s_1)][(s_2,s_1),\texttt{\#}][\texttt{\#},(p_3,p_2)][(p_3,p_2),r_2][r_2,s_2][s_2,\texttt{\#}][\texttt{\#},p_3][p_3,(r_3,r_2)]\\
&&[(r_3,r_2),s_2][s_2,\texttt{\#}][\texttt{\#},p_3][p_3,r_3][r_3,(s_3,s_2)]
\end{eqnarray*}
is a valid computation of $M$.
\end{example}

The length of a valid computation can be easily calculated.

\begin{lemma}\label{lem:lengthvalc}
Let $M$ be an $\oca$ on input $w_1w_2 \cdots w_n$ which
is accepted after $t$ time steps. Then the length
of the corresponding valid computation is $n+(n+1)\cdot n \cdot t$.
\end{lemma}

The next lemma is the key tool for the reductions.

\begin{lemma}\label{lem:expovalc}
Let $M$ be an $\oca$. Then language 
\[
\valc\,'(M)=\{\,\dollar_1 x_1 x_2 \cdots x_m \dollar_2 a^{2^0} bb a^{2^1} bb \cdots bba^{2^{m-1}} bb \kand
\mid m \ge 1, x_1x_2 \cdots x_m \in \valc(M)\,\}
\]
belongs to $\lfam(\fdrcpcfa{O(\log(n))}(4))$. 
\end{lemma}

\begin{proof} 
In~\cite{bordihn:2011:uhrpcfa} a $\fdrcpcfa{O(n)}(3)$ is constructed
that accepts $\valc(M)$.
Basically, the master component $A_1$ and component $A_2$ are used to verify that
after every subconfiguration the correct successor subconfiguration is given,
whereas component $A_3$ is used to check the correct format of the input.
This construction can be implemented identically for the present construction
if we interpret $\dollar_2$ as the right endmarker.
Additionally, component $A_4$ is used in the same way as component $A_3$ in
the construction of Lemma~\ref{lem:expoww}, that is, initially it
reads $\dollar_1$ and $x_1$, stores $x_1$ in its state,
and waits at position 2 until it is queried. After being reset to its
initial state, it again reads the next input symbol, stores it, and waits.

When $x_1x_2 \cdots x_m \in \valc(M)$ is tested, the master $A_1$ and component
$A_2$ are both located at $\dollar_2$. The second part of the input is now
tested along the line of the construction given in the proof of
Lemma~\ref{lem:expoww}, where the master plays the role of the master,
component $A_2$ the role of component $A_2$, and component~$A_4$ the role
of component $A_3$.

The length of a word $w\in \valc\,'(M)$ is
$|w|=3m+3+\sum_{i=0}^{m-1}2^i=2^{m}+3m+2$,
for some $m\geq 1$. 
The test whether $x_1x_2 \cdots x_m$ belongs to $\valc(M)$ requires $O(m)$
communications. For the remaining tests additional $O(m)$ communications
are necessary as the proof of Lemma~\ref{lem:expoww} shows.
So, altogether, $O(m)$ communications are sufficient
which is of order $O(\log(|w|))$.
\end{proof} 

The set of \emph{invalid computations} $\invalc\,'(M)$ 
is simply defined to be the complement of $\valc\,'(M)$ with respect to the alphabet
$\{a,b,\dollar_1,\dollar_2,\kand\} \cup (A \times A)$. 

\begin{lemma}\label{lem:expoinvalc}
Let $M$ be an OCA. Then language 
$\invalc\,'(M)$ belongs to 
$\lfam(\fdrcpcfa{O(\log(n))}(4))$. 
\end{lemma}

\begin{proof}
To accept the set of invalid computations $\invalc\,'(M)$ almost the same construction
as for Lemma~\ref{lem:expovalc} can be used. The only adaption concerns
acceptance and rejection. Since the only possibility to accept is that the
master halts in state $accept$ while the other components are non-halting, 
accepting computations can be made rejecting
by sending the master into a halting non-accepting state $reject$ instead.
In order to make rejecting computations accepting, it is now sufficient
to send the components into some halting accepting state whenever
they would halt rejecting.
\end{proof}

\begin{theorem}\label{thm:dec} 
For any degree $k\geq 4$, emptiness, finiteness, infiniteness,
universality, inclusion, equivalence, regularity, and context-freeness
are not semi-decidable for $\fdrcpcfa{O(\log(n))}(k)$.  
\end{theorem}

\begin{proof} 
All these problems are known to be non-semi-decidable for
$\oca$~\cite{Malcher:2002:dccadq}. By standard techniques (cf., for example,~\cite{bordihn:2011:uhrpcfa})
the $\oca$ problems are reduced to $\fdrcpcfa{O(\log(n))}(k)$
via the valid and invalid computations and Lemmas~\ref{lem:expovalc}
and~\ref{lem:expoinvalc}.
\end{proof}

\subsection{Undecidability of Communication Boundedness}

This subsection is devoted to questions concerning the decidability
or computability of the communication bounds. In principle,
we deal with three different types of problems. The first
type is to decide for a given $\drcpcfa(k)$ $A$
and a given function $f$ whether or not $A$ is communication 
bounded by $f$. The next theorem solves this problem negatively for 
all non-trivial communication bounds and all degrees~$k \ge 3$.

\begin{theorem}\label{thm:undec:comm}
Let $k \ge 3$ be any degree,  $f \in o(n)$, and $A$ be a $\drcpcfa(k)$. 
Then it is not semi-decidable whether
$A$ is communication bounded by $f$.
\end{theorem}

\begin{proof}
Let $A$ be a $\drcpcfa(k)$ with $k \ge 3$ accepting some language $L(A) \subseteq \Sigma^*$.
We take two new symbols $\{a,\dollar\} \cap \Sigma=\emptyset$ and construct a $\drcpcfa(k)$ $A'$
accepting language $a^* \dollar L(A)$. The idea of the construction is that, initially, all components
move synchronously across the leading $a$-block. During this phase, the master
component queries one of the non-master components in every time step.
When all components have read the separating symbol $\dollar$, they enter 
the initial state of the corresponding component of $A$. 
Subsequently, $A$ is simulated, thus testing whether the remaining input belongs to $L(A)$.
So, on input $a^n \dollar w$ with $n \ge 1$ and $w \in L(A)$, $A'$ performs at
least~$n$ communications. In particular, for $n\geq |w|$ we obtain words that
show that $A'$ is not communication bounded by any function $f \in o(n)$,
unless $L(A)$ is empty. 
So, $A'$ is a $\fdrcpcfa{f}(k)$ if and only if $L(A)=\emptyset$.
 
Since in~\cite{bordihn:2011:uhrpcfa} it has been shown that 
emptiness is not semi-decidable for $\drcpcfa$ with at least three components,
the theorem follows.
\end{proof}

Mitrana considers in~\cite{mitrana:2000:odcipcfas} the \emph{degree of
communication} of parallel communicating
finite automata systems. The degree of communication of an accepting computation is defined as the
number of queries posed. The degree of communication $Comm(x)$ of a 
nondeterministic $\pcfa$ $A$ on input $x$ is defined as the minimal 
number of queries posed in accepting computations on $x$.
The degree of communication $Comm(A)$ of a $\pcfa$ $A$ is then defined 
as $\sup \{\, Comm(x) \mid x \in L(A) \,\}$. 
Here we have the second type of problems we are dealing with.
Mitrana raised the question whether the degree of communication $Comm(A)$
is computable for a given nondeterministic $\pcfa(k)$ $A$. Since $Comm(A)$ is either
finite or infinite, in our terms the question is to decide
whether or not $A$ is communication bounded by some function $f\in O(1)$
and, if it is, to compute the precise constant. 
The next theorem solves the problem.

\begin{theorem}\label{theo:undec:comm1}
Let $k\geq 3$ be an integer. Then the degree of communication $Comm(A)$ is not 
computable for $\drcpcfa(k)$. 
\end{theorem}

\begin{proof}
For a given $\drcpcfa(k)$ $A$ and new input symbols $a$ and $\dollar$, we
construct a $\drcpcfa(k)$ $A'$ accepting the language $a^* \dollar L(A)$
as in the proof of Theorem~\ref{thm:undec:comm}.

Now, we claim that $Comm(A')=0$ if and only if $L(A)=\emptyset$.
If $L(A)$ is empty, then $A'$ accepts the empty set and, thus, $Comm(A')=0$. 
On the other hand, if $L(A)$ is not empty, then $Comm(A') > 0$ by construction 
of $A'$. Since emptiness is not semi-decidable for $\drcpcfa(k)$ with 
$k\geq 3$~\cite{bordihn:2011:uhrpcfa}, the theorem follows.
\end{proof} 

Now we turn to the last type of problems we are dealing with
in this section. The question is now whether the degree of 
communication is computable for the \emph{language} accepted
by a given nondeterministic $\pcfa(k)$ $A$.
In~\cite{mitrana:2000:odcipcfas} the degree of
communication $Comm_{X}(L)$ of a language $L$
is defined as $\inf \{\, Comm(A) \mid A \text{ is device of type } X \text{ and } L(A)=L \,\}$. 
Mitrana showed in~\cite{mitrana:2000:odcipcfas} that 
$Comm_{CPCFA}(L(A))$ for some nondeterministic $\cpcfa$ $A$ is not
computable. 
He leaves as an open question whether the degree is computable for $\rcpcfa$.
Here we are going to show that the degree is not even computable for
deterministic $\rcpcfa$.

\begin{lemma}\label{lem:undec:comm2}
Let $k\geq 3$ be an integer. Then the degree of communication $Comm_{DRCPCFA(k)}(L(A))$ is not 
computable. 
\end{lemma}

\begin{proof}
For a given $\drcpcfa(k)$ $A$ over alphabet $\Sigma$ and new input 
symbols $b,0,1,\dollar_1,\dollar_2$, we
construct a $\drcpcfa(k)$ $A'$ accepting the language 
\[
\{\, w_1w_2 \cdots w_m b w_1w_2 \cdots w_m \mid m \ge 1, w_i \in \{0,1\}, 1
\le i \le m\,\} \dollar_1 \dollar_2 L(A).
\]
We present the construction for $k=3$. The generalization
to larger $k$ is straightforward. 

The idea of the construction is that in a first phase master component $A_1$ and 
a non-master component~$A_2$ check the correctness of the prefix 
$w_1w_2 \cdots w_m b w_1w_2 \cdots w_m$. This is done as in the construction
of Lemma~\ref{lem:wbw}. Component $A_3$ checks the correct format of the
input up to the separating symbol $\dollar_1$ and waits on symbol
$\dollar_2$ until it is queried.
At the end of this phase, the master is on the $\dollar_1$ and component $A_2$
stays on the symbol $b$.

In a second phase, the master component stays on $\dollar_1$ and repeatedly queries
component $A_2$ until this one has read $\dollar_1$ and, thus, stays on
$\dollar_2$. Now the master reads $\dollar_1$ and queries component $A_2$.
After being reset to its initial state, component $A_2$ reads $\dollar_2$
and performs one $\lambda$-step. Then it changes to the initial state
of $A_2$ in $A$. During this $\lambda$-step, the master component
reads $\dollar_2$ and queries component $A_3$. Then it changes
to the initial state of the master of $A$. Finally, after being reset to
its initial state, component $A_3$ reads $\dollar_2$ and changes
into the initial state of $A_3$ in $A$.

Now, all components are in their initial states on the first symbol 
of the input of $A$ and in a third phase~$A$ is simulated. 
We claim that $Comm(L(A'))=0$ if and only if $L(A)=\emptyset$. 
If $L(A)$ is empty, then $A'$ accepts the empty set and
$Comm(L(A'))=Comm(\emptyset)=0$.
If $L(A)$ is not empty, we fix some $x \in L(A)$. 
Assume contrarily that $Comm(L(A'))=0$. Then there exists
a $\drcpcfa(k)$ $B$ accepting $L(A')$ such that
$Comm(B)=0$. From $B$ a $\drcpcfa(k+1)$ $B'$ is constructed by
providing an additional component which checks whether the suffix
is precisely $x$, and halts non-accepting if an error is found.
So, $B'$ accepts the language  
\[
\{\, w_1w_2 \cdots w_m b w_1w_2 \cdots w_m \mid m \ge 1, w_i \in \{0,1\}, 1
\le i \le m\,\} \dollar_1 \dollar_2 x
\] 
and we still have $Comm(B')=0$. 
Similar as in the proof of Lemma~\ref{lem:wbw}, it follows by 
an incompressibility argument that this conclusion
leads to a contradiction.

Since emptiness is not semi-decidable for $\drcpcfa(k)$ with 
$k\geq 3$~\cite{bordihn:2011:uhrpcfa}, the degree of communication $Comm_{DRCPCFA(k)}(L(A))$ is not 
computable. 
\end{proof}

\section{An Infinite Hierarchy}\label{sec:hierarchy}

In this section, we are going to show that there is an infinite strict hierarchy of language
classes in between $\lfam(\fdrcpcfa{O(\log(n))}(k))$ and
$\lfam(\fdrcpcfa{O(\sqrt{n})}(k))$, for any $k\geq 4$.
To this end, we consider functions $f:\mathbb{N}\to\mathbb{N}$
that are time-computable by one-way cellular automata. 
That means, given any unary input of length $n\geq 1$, say $a^n$, the rightmost cell 
has to enter an accepting state exactly after~$f(n)$ time
steps and never before. 
Time-computable functions in $\oca$ have been 
studied in~\cite{Buchholz:1998:tcfoca:art}, where it is shown that,
for any $r\geq 1$, there exists an $\oca$-time-computable function
$f\in \Theta(n^r)$. We will use this result in the sequel.
So, let $M_r$ be an $\oca$ that time-computes $f\in \Theta(n^r)$,
for $r\geq 1$. We will use 
\begin{multline*}
L_r =\{\, \dollar_1x_1x_2\cdots x_\ell \dollar_2 w'_1w'_2 \cdots w'_m w_{m+1}\cdots  w_\ell
\dollar_3 w'_1w'_2 \cdots w_m' w_{m+1}\cdots  w_\ell \dollar_4 
a^{2^{0}}bb a^{2^{1}} bb \cdots a^{2^{m-1}} bb \kand \mid\\ 
m\geq 1, x_1x_2\cdots x_\ell \text{ is the valid computation of } M_r \text{
  on input }a^m,\\ \, w_i' \in \{0',1'\}, 1 \le i \le m,\,  w_i \in \{0,1\}, m+1 \le i \le \ell \,\}
\end{multline*}
as witness languages for the infinite hierarchy.

\begin{lemma}\label{lem:inf:hier1}
Let $r\geq 1$ be an integer. Then language 
$L_r$ belongs to $\lfam(\fdrcpcfa{O(\log(n)^{r+2})}(4))$.
\end{lemma}

\begin{proof}
An $\fdrcpcfa{O(\log(n)^{r+2})}(4))$ $A$ accepting $L_r$ works in five phases.

As mentioned before, in~\cite{bordihn:2011:uhrpcfa} an $\fdrcpcfa{O(n)}(3)$ is constructed
that accepts $\valc(M)$, where the master component $A_1$ and component 
$A_2$ are used to verify the subconfigurations, and component $A_3$ is 
used to check the correct format of the input.
In the first phase, $A$ simulates this behavior where $\dollar_2$
plays the role of the endmarker.
When $x_1x_2 \cdots x_\ell \in \valc(M)$ has been tested, the master $A_1$ and 
component~$A_2$ are both located on the symbol after $\dollar_2$, that is, on $w_1'$.
Additionally, component $A_4$ initially reads $\dollar_1$ and waits on $x_1$
to be queried. The total number of communications in this phase is of order $O(\ell)$.

In the second phase, it is verified that there are as many symbols in between
$\dollar_1$ and $\dollar_2$ as in between~$\dollar_2$ and $\dollar_3$, that
is, the length $\ell$ is matched. Furthermore, it is checked whether there
are exactly~$m$ symbols of the second infix primed.
Since $x_1x_2\cdots x_\ell$ describes an $\oca$ computation on some 
unary input $a^m$, the initial configuration of the $\oca$ is of the form
$\#(a')^m$. Therefore, the valid computation 
begins with $[\texttt{\#},a'][a',a']^{m-1}[a',\texttt{\#}]$ followed by
symbols not containing primed versions of other symbols.
As in the constructions before, the master $A_1$ moves
to the right while querying component $A_4$ in every step. Whenever
component $A_4$ is reset to its initial state, it reads the next 
input symbol, remembers it, and waits. 
In this way, component $A_4$ is tracked over the valid computations.
Moreover, the master $A_1$ receives information about the symbols read by
$A_4$ and can check the number of primed symbols to be~$m$.
The phase ends successfully when $A_1$ has read $\dollar_3$ and 
receives the information
that $A_4$ has read $\dollar_2$ in this moment, that is, both infixes
have the same length $\ell$.
This phase takes $O(\ell)$ communications. At its end, the master 
$A_1$ is located on the symbol 
after $\dollar_3$ and components $A_2$ and $A_4$ are both located on 
the symbol after $\dollar_2$.

The third phase is used to compare the word in between $\dollar_2$ and
$\dollar_3$ with the word in between $\dollar_3$ and $\dollar_4$.
Similar as in the phase before, to this end, the master $A_1$ moves
to the right while querying component $A_2$ in every step. Whenever
component $A_2$ is reset to its initial state, it reads the next 
input symbol, remembers it, and waits. So, $A_1$ can check whether
the currently read symbols are identical. The phase ends successfully
when $A_1$ has read $\dollar_4$ and receives the information
that $A_2$ has read $\dollar_3$ in this moment. Now,
the master $A_1$ is located on the symbol 
after $\dollar_4$, $A_2$ is located on the symbol 
after $\dollar_3$, and~$A_4$ still on the symbol 
after $\dollar_2$. 
The total number of communications in this phase is of order $O(\ell)$.

The fourth phase is used to track component $A_2$ to the position of $A_1$.
So, the master $A_1$ loops on its position while it queries $A_2$ in every
step. In this way, $A_2$ moves to the right. The phase ends when 
$A_1$ receives the information that $A_2$ has read $\dollar_4$.
At this time step the master $A_1$ and component $A_2$ are located on the symbol 
after $\dollar_4$ and $A_4$ still on the symbol 
after $\dollar_2$. During this phase $O(\ell)$ communications take place. 

The fifth and final phase is to check the suffix.
The master knows that this phase starts and changes into some appropriate 
state in a $\lambda$-step. The situation is similar for component $A_2$.
It is in its initial state on a symbol $a$ for the first time. So, both
synchronously start the phase. Basically, here we can use again the
construction of the proof of Lemma~\ref{lem:expovalc}. That is, the master component
and component $A_2$ check that the lengths of $a$-blocks are doubling. 
Communication takes place at both symbols $b$. 
Reading the first $b$, component $A_4$ is queried and forced to proceed
one input symbol in order to check the correct number $m$ of $a$-blocks. 
Since component $A_4$ is tracked over an infix whose first $m$ symbols
are primed this can be done almost as before.
Reading the second $b$, the master queries component $A_2$
to ensure that the $a$-blocks ended correctly.
The total number of communications in this phase is of order $O(m)$.
This concludes the construction of $A$.

The length $\ell$ of the valid computation of $M_r$ on input $a^m$
is of order $\Theta(m^2\cdot m^r)=\Theta(m^{r+2})$ by Lemma~\ref{lem:lengthvalc}.
The length of an input is $n=3\ell+2^m-1+2m+5 \in \Theta(2^m)$. 
The total number of communications is of order 
$O(\ell)+O(\ell)+O(\ell)+O(\ell)+O(m)= O(m^{r+2})$.
So, the number of communications is of order $O(\log(n)^{r+2})$.
\end{proof}

\begin{lemma}\label{lem:inf:hier2}
Let $r\geq 1$ be an integer. Then language 
$L_r$ does not belong to $\lfam(\fdrcpcfa{O(\log(n)^{r})}(4))$.
\end{lemma}

\begin{proof}
The proof is along the line of the proof of Lemma~\ref{lem:polyww}.
By way of contradiction, we assume that $L_r$ is accepted by some
$\fdrcpcfa{O(\log(n)^{r})}(4)$. 

Let $z$ be a word in $L_r$ whose infix $x=x_1x_2\cdots x_\ell$
is the valid computation of $M_r$ on input $a^m$. 
Then $|z|$ is of order $\Theta(2^{m})$ and $\ell$ is of order $\Theta(m^{r+2})$.
We will use an incompressibility argument and choose a string
$w=w_1w_2\cdots w_\ell \in\{0,1\}^*$ so that the Kolmogorov
complexity is $C(w)\geq |w|=\ell\in \Theta(m^{r+2})$. Then the word
$z' = \dollar_1 x \dollar_2 w'_1w'_2 \cdots w'_m w_{m+1}\cdots  w_\ell
\dollar_3 w'_1w'_2 \cdots w_m' w_{m+1}\cdots  w_\ell \dollar_4 
a^{2^{0}}bb a^{2^{1}} bb \cdots a^{2^{m-1}} bb \kand$
belongs to~$L_r$ as well.

With the help of the accepting computation on $z'$ we write down 
a program that uniquely reconstructs~$w$.
The order of magnitude of the size of the program is given by the product of the size 
of one snapshot and the number of all snapshots. Since one snapshot 
can be described by $O(m)$ bits and the number of snapshots is bounded 
by $O(m^{r})$, we derive that $C(w)$ is of order $O(m^{r+1})$,
a contradiction.
\end{proof}

Combining Lemma~\ref{lem:inf:hier1} and Lemma~\ref{lem:inf:hier2}
the desired infinite hierarchy of the next theorem follows.

\begin{theorem}
Let $r\geq 1$ be an integer. Then
the class $\lfam(\fdrcpcfa{O(\log(n)^{r})}(4))$ is properly included
in the class $\lfam(\fdrcpcfa{O(\log(n)^{r+2})}(4))$.
\end{theorem}

Since the proofs of Lemma~\ref{lem:inf:hier1} and Lemma~\ref{lem:inf:hier2} do
not rely on a specific number of components as long as at least four components are provided,
the hierarchy follows for any number of components $k\geq 4$.

\begin{corollary}
Let $k\geq 4$ and $r\geq 1$ be two integers. Then
the class $\lfam(\fdrcpcfa{O(\log(n)^{r})}(k))$ is properly included
in the class $\lfam(\fdrcpcfa{O(\log(n)^{r+2})}(k))$.
\end{corollary}

\end{document}